\documentclass[12pt]{article}
\usepackage[cm]{fullpage}

\usepackage{amsmath,amssymb,amsthm}
\usepackage{pgfplots}
\usepackage{comment}
\usepackage[textsize=footnotesize,color=green!40]{todonotes}
\usepackage[utf8]{inputenc}
\usepackage{amsfonts}
\usepackage{amssymb}
\usepackage{mdframed}
\usepackage{framed}
\usepackage{hyperref}
\usepackage{mathtools}
\usepackage{authblk}
\usepackage{fancyhdr}
\usepackage{float}

\usepackage[noend]{algorithmic}
\usepackage[algoruled,vlined,nofillcomment,algo2e]{algorithm2e}
\usepackage{bm}

\usepackage{tikz}
\usetikzlibrary{calc,backgrounds}
\usetikzlibrary{arrows,calc,decorations.pathreplacing,positioning,shapes,shapes.geometric,shapes.callouts}
\usetikzlibrary{shapes,intersections,external}

\newcommand{\mech}[1]{\mathcal{#1}}
\newcommand{\tup}[1]{\vec{#1}}
\newcommand{\rv}[1]{\mathsf{#1}}

\newcommand{\trv}[1]{\tup{\rv{#1}}}

\renewcommand{\Pr}{\mathbb{P}}
\newcommand{\Ex}{\mathbb{E}}
\newcommand{\Var}{\mathbb{V}}

\newcommand{\PP}{\mathbb{P}}
\newcommand{\TV}{\mathfrak{T}}
\renewcommand{\TV}{\mathrm{TV}}

\newcommand{\remove}[1]{}

\usepackage{amsthm}

\newtheorem{theorem}{Theorem}[section]
\newtheorem{lemma}{Lemma}[section]
\newtheorem{definition}{Definition}[section]
\newtheorem{corollary}{Corollary}[section]

 \date{}

\title{Improved Summation from Shuffling}
\author{Borja Balle\thanks{{\tt borja.balle@gmail.com}.}~~~~~James Bell\thanks{The Alan Turing Institute. {\tt jbell@turing.ac.uk}. Work supported by the UK Government’s Defence \& Security Programme in support of the Alan Turing Institute.}~~~~~Adria Gascon\thanks{Google. {\tt adriagascon@gmail.com}. Work partly done when A.G was at The Alan Turing Institute and Warwick University, and  supported by The Alan Turing Institute under the EPSRC
grant EP/N510129/1, and the UK Government’s Defence \& Security Programme in support of the Alan Turing Institute.}~~~~~Kobbi Nissim\thanks{Dept.\ of Computer Science, Georgetown University. {\tt kobbi.nissim@georgetown.edu}. Work supported by NSF grant no.~1565387, 
TWC: Large: Collaborative: Computing Over Distributed Sensitive Data. Work partly done while K.~N.\ was visiting the Alan Turing Institute.}}

\begin{document}

\maketitle

\begin{abstract}
  A protocol by Ishai et al.\ (FOCS 2006) showing how to implement distributed $n$-party summation from secure shuffling has regained relevance in the context of the recently proposed \emph{shuffle model} of differential privacy, as it allows to attain the accuracy levels of the curator model at a moderate communication cost. To achieve statistical security $2^{-\sigma}$, the protocol by Ishai et al.\ requires the number of messages sent by each party to {\em grow} logarithmically with $n$ as $O(\log n + \sigma)$. In this note we give an improved analysis achieving a dependency of the form $O(1+\sigma/\log n)$. Conceptually, this addresses the intuitive question left open by Ishai et al.\ of whether the shuffling step in their protocol provides a ``hiding in the crowd'' amplification effect as $n$ increases. From a practical perspective, our analysis provides explicit constants and shows, for example, that the method of Ishai et al.\ applied to summation of $32$-bit numbers from $n=10^4$ parties sending $12$ messages each provides statistical security $2^{-40}$.
\end{abstract}

\section{Introduction}
\label{sec:intro}

Ishai et al.~\cite{ikos} showed how to use anonymous communications as a building block for a variety of tasks, including secure computation of $n$-party summation. In the setting of Ishai et al., $n\geq 2$ users hold values $x_1,\ldots,x_n$ in $\mathbb{Z}_m$ and wish to reveal to a server the sum of their values (and nothing else) by simultaneously sending anonymous messages to the server. A naive solution would require each user $i$ to anonymously send $x_i$ messages to the server, which then counts the total number of received messages. The solution by Ishai et al.\ is much more efficient: each user $i$ splits their input $x_i$ into $k$ additive shares, and sends all shares anonymously to the server. The server then obtains $nk$ shares and reconstructs the result by adding them up. Surprisingly, it is shown in \cite{ikos} that
$k = O(\log m + \sigma + \log n)$
(see~\cite{DBLP:journals/corr/abs-1906-09116} for explicit, small constants)
suffices to achieve statistical security $\sigma$, in the sense that
the server the set of shares submitted by the users cannot distinguish two inputs $(x_1,\ldots,x_n)$ and $(y_1,\ldots,y_n)$ with $\sum_i x_i = \sum_i y_i$, except with advantage $2^{-\sigma}$.
We refer informally to the protocol from Ishai et al.\ as the {\em IKOS protocol}.

The work of Ishai et al.\ is of significant relevance in the context of the recently proposed {\em shuffle model} of Differential Privacy (DP)~\cite{BEMMRLRKTS17} (see also \cite{erlingsson2019amplification,DBLP:journals/corr/abs-1808-01394}). In this model, a trusted shuffler applies a random permutation to messages sent by users before these are received by a server. Note that this setting is essentially the same as the one of Ishai et al., as the shuffler provides an anonymous communication channel. This connection has been shown recently in independent works by Balle et al.~\cite{DBLP:journals/corr/abs-1906-09116} and Ghazi et al.~\cite{DBLP:journals/corr/abs-1906-08320} which, by a black-box application of the IKOS protocol, showed that in the shuffle model one can achieve the same accuracy/privacy trade-offs than in the central model of DP for the task of real summation. Crucially, this can be done while only requiring $O(\log n)$ messages per party, and this is enabled by the IKOS protocol. In fact, Ghazi et al.~\cite{DBLP:journals/corr/abs-1906-08320} had reinvented the protocol specifically for this application. These results improve on previous work by Cheu et al.~\cite{DBLP:journals/corr/abs-1808-01394}, where it was shown that the same goal can be achieved with $O(\sqrt{n})$ messages per party. 

An intriguing question was left open in the work of Ishai et al.: while the IKOS protocol works for every number of users $n \geq 2$, one could hope to improve the dependency of the number of required messages on $n$. Intuitively, a larger number of participants should allow a single user $i$'s contribution to ``hide in the crowd'' hence enabling splitting $x_i$ into a number of shares that does not increase with $n$. This question is of special relevance in the context of applications to the shuffle model of differential privacy, as a relatively large $n$ is required due to privacy-accuracy tradeoffs.

In this paper we resolve the above question positively, showing that larger $n$ does indeed help to reduce the number of required messages. Concretely, we show that for fixed $n, m$ and security parameter $\sigma$, it suffices to take the number of shares to be $k = \left\lceil\frac{2\sigma+\log_2(m)}{\log_2(n)-\log_2(e)}+2\right\rceil$. In particular, for securely computing, with security $\sigma=40$, the sum of $32$-bit numbers by $n=10^4$ users, our bounds show that $12$ messages per party are enough. Furthermore, one of these messages can be sent in the clear.

We recently learned that Ghazi, Manurangsi, Pagh, and Velingker have, independently of our work, obtained an analysis of the IKOS protocol that provides guarantees for a constant number of messages but we are unaware of the specifics.
 
\section{Preliminaries}
\label{sec:prelims}
We will denote the additive group that we wish to sum in by $\mathbb{G}$. This group was taken to be $\mathbb{Z}_m$ in the example in the introduction but is only required to be an abelian group of size $m$. As in the introduction we have $n$ users each holding a private value $x_i\in \mathbb{G}$ and we denote the tuple of these values as $\tup{x}=(x_1,\ldots ,x_n)$. We denote random variables by upper case letters $\rv{X},\rv{Y},\ldots$ and tuples of them $\tup{\rv{X}},\tup{\rv{Y}},\ldots$. We denote randomized maps $\mech{R},\mech{S},\ldots$.

For $k\geq 2$ and $x \in \mathbb{G}$ we define the \emph{$k$-additive sharing} of $x$ as the output of the randomized map $\mech{R}_k : \mathbb{G} \to \mathbb{G}^k$ given by $\mech{R}_k(x) = (\rv{Y}^{(1)}, \ldots, \rv{Y}^{(k)})$, where $\rv{Y}^{(j)}$, $j \in [k]$, is a tuple of uniformly random group elements (i.e.\ \emph{shares}) conditioned on $\sum_j \rv{Y}^{(j)} = x$.

We use the notation $\Pr$ for probability, $\Ex$ for expectation, $\Var$ for variance, and $\TV$ for total variation distance.

\section{Our results}
\label{sec:results}

We define the \emph{$k$-parallel IKOS protocol} with $n$ users as the randomized map $\mech{V}_{k,n} : \mathbb{G}^n \to (\mathbb{G}^n)^k$ obtained as follows.
Let $\mech{S}^{(j)} : \mathbb{G}^n \to \mathbb{G}^n$, $j \in [k]$, be $k$ independent shufflers returning a uniform random permutation of their inputs.
For any $\tup{x} = (x_1,\ldots,x_n) \in \mathbb{G}^n$ define the random variables $(\rv{Y}^{(1)}_i, \ldots, \rv{Y}^{(k)}_i) = \mech{R}_k(x_i)$, $i \in [n]$.
Then, the IKOS protocol returns, for $j \in [k]$, the result of independently shuffling the $j$th shares of all the users together:
\begin{align}
\mech{V}_{k,n}(\tup{x}) = \left(\mech{S}^{(1)}(\rv{Y}^{(1)}_1, \ldots, \rv{Y}^{(1)}_n), \ldots, \mech{S}^{(j)}(\rv{Y}^{(j)}_1, \ldots, \rv{Y}^{(j)}_n), \ldots, \mech{S}^{(k)}(\rv{Y}^{(k)}_1, \ldots, \rv{Y}^{(k)}_n)\right) \enspace \label{eq:ikos}
\end{align}
Whenever $n$ or $k$ (or both) are clear from the context we omit them to unclutter our notation.

The IKOS protocol provides a method to compute the sum of the users' inputs $\sum_{i \in [n]} x_i$ by summing all the group elements in the result $\mech{V}_{k,n}(\tup{x}) \in (\mathbb{G}^n)^k$.
This follows from observing that such sum is just a sum of the shares of every user.
We define the statistical security of the parallel IKOS protocol in the standard way. That is, we require the output distributions to be indistinguishable whenever the protocol is executed on two inputs $\tup{x}, \tup{x}' \in \mathbb{G}^n$ such that $\sum_i x_i = \sum_i x_i'$.
As usual, indistinguishability is measured in terms of total variation (i.e.\ statistical) distance.

This is made precise in the following definitions, where we consider the cases with both fixed and random inputs:
\begin{enumerate}
\item We say that a protocol $\mech{V}$ provides \emph{worst-case statistical security} with parameter $\sigma$ if for any $\tup{x}, \tup{x}' \in \mathbb{G}^n$, such that $\sum_{i \in [n]} x_i = \sum_{i \in [n]} x_i',$ we have $\TV(\mech{V}(\tup{x}),\mech{V}(\tup{x}')) \leq 2^{-\sigma}$.
\item We say that a protocol $\mech{V}$ provides \emph{average-case statistical security} with parameter $\sigma$ if we have $\Ex_{\tup{\rv{X}},\tup{\rv{X}}'}[\TV_{|\tup{\rv{X}},\tup{\rv{X}}'}(\mech{V}(\tup{\rv{X}}),\mech{V}(\tup{\rv{X}}'))] \leq 2^{-\sigma}$, where $\tup{\rv{X}}$ and $\tup{\rv{X}}'$ are $n$-tuples of uniform random elements from $\mathbb{G}$ conditioned on $\sum_{i \in [n]} \rv{X}_i = \sum_{i \in [n]} \rv{X}_i'$.
\end{enumerate}

The following theorem states our main technical result.

\begin{theorem}[Average-case security]\label{thm:avg-security}
The \emph{$k$-parallel IKOS protocol} with $k\geq 3$ and $n\geq 19$ users provides average-case statistical security with distance $2^{-\sigma}$, for
  \begin{equation}
    \sigma=\frac{(k-1)(\log_2(n)-\log_2(e))-\log_2(m)}{2}
  \end{equation}
  provided $\sigma\geq 1$.
\end{theorem}

While the above theorem only states average-case security, a simple randomization trick recovers worst-case security at the cost of one extra message per party.
Moreover, such a message does not need to be shuffled. This corresponds to a small variation on the parallel IKOS protocol where one of the messages contributed by each user is not sent through a shuffler; i.e.\ it is possible to unequivocally associate one of the messages from the input to each user.
We define the \emph{$k$-parallel IKOS with randomized inputs protocol} as the randomized map $\mech{V}_{k,n} : \mathbb{G}^n \to (\mathbb{G}^n)^{k+1}$ obtained as follows.
Let $\mech{S}^{(j)} : \mathbb{G}^n \to \mathbb{G}^n$, $j \in [k]$, be $k$ independent shufflers returning a uniform random permutation of their inputs.
For any $\tup{x} = (x_1,\ldots,x_n) \in \mathbb{G}^n$ define the random variables $(\rv{Y}^{(1)}_i, \ldots, \rv{Y}^{(k+1)}_i) = \mech{R}_{k+1}(x_i)$, $i \in [n]$, obtained by sampling $k+1$ additive shares for each input.
Then, the IKOS with randomized inputs protocol returns, for $j \in [k]$, the result of independently shuffling the $j$th shares of all the users together, concatenated with the $k+1$th \emph{unshuffled} shares:
\begin{align}
\tilde{\mech{V}}_{k,n}(\tup{x}) = \left(\mech{S}^{(1)}(\rv{Y}^{(1)}_1, \ldots, \rv{Y}^{(1)}_n), \ldots, \mech{S}^{(j)}(\rv{Y}^{(j)}_1, \ldots, \rv{Y}^{(j)}_n), \ldots, \mech{S}^{(k)}(\rv{Y}^{(k)}_1, \ldots, \rv{Y}^{(k)}_n), (\rv{Y}^{(k+1)}_1, \ldots, \rv{Y}^{(k+1)}_n)\right) \enspace.
\label{eq:ikos-half}
\end{align}

\begin{corollary}\label{cor:wst-security}
The \emph{$k$-parallel IKOS with randomized inputs protocol} with $k\geq 3$ and $n\geq 19$ users provides worst-case statistical security with distance $2^{-\sigma}$,
with $\sigma$ given by the same expression as in Theorem~\ref{thm:avg-security}.
  Thus, for fixed $n,m$ and $\sigma$, it suffices, for worst-case security, to take the number of shuffled messages to be
  \begin{equation}
    k=\left\lceil\frac{2\sigma+\log_2(m)}{\log_2(n)-\log_2(e)}+1\right\rceil \enspace.
  \end{equation}
\end{corollary}

Using the analysis in \cite{DBLP:journals/corr/abs-1906-09116}, $k$ is required to be at least $2\sigma $ and then grows with $n$ and $m$, so is not constant in $n$ and for $\sigma=40$ is always at least $80$. In this work the $2\sigma $ and the $\log(m)$ are divided by a logarithmic factor of $n$. Thus, for fixed $\sigma$, so long as $m$ grows at most polynomially in $n$, the number of required messages is bounded by a constant. Further, for reasonable parameter values, $k$ is much less than $80$ and, asymptotically, if $m= O(n^{2-\epsilon})$ it converges to $3$.
 
\section{The proof}
\label{sec:the_proof}

In this section we give our full proof of Theorem~\ref{thm:avg-security} and Corollary~\ref{cor:wst-security}. We start, in Section~\ref{sec:proof_outline}, by proving the Theorem assuming the Lemmas that are proved in Sections~\ref{sec:single_shuffling}, \ref{sec:graph_argument}, and \ref{sec:bound_on_graph}. Finally in Section~\ref{sec:random_inputs} we prove the Corollary.

\subsection{Proof Outline}
\label{sec:proof_outline}
The following proof is of Theorem~\ref{thm:avg-security}. In this proof we will provide forward references to the required lemmas which are then proved in the rest of this section.

\begin{proof}[Proof of Theorem~\ref{thm:avg-security}]
  Lemma~\ref{lem:single_shuffling} in Section~\ref{sec:single_shuffling} says that this protocol provides statistical security with distance bounded by the following expression, which is given here in terms of an event $E$ specified in Section~\ref{sec:single_shuffling}.
  \begin{equation*}
    \sqrt{m^{kn-1} \Pr[E] - 1}
  \end{equation*}
  In Section~\ref{sec:graph_argument} we define a distribution over multigraphs. Lemma~\ref{lem:graph_argument} says that if $G$ is drawn from this distribution and $C(G)$ is the number of connected component in $G$ then,
  \begin{equation*}
    \Pr[E]\leq m^{-kn}\Ex[m^{C(G)}]\enspace.
  \end{equation*}
  This expectation is then bounded in Lemma~\ref{lem:bound_on_graph}, which says that, if $n\geq 19$, $k\geq 3$ and $m\leq \frac{1}{2}\left(\frac{n}{e}\right)^{k-1}$, 
  \begin{equation*}
    \Ex[m^{C(G)}]\leq m+m^2\left(\frac{n}{e}\right)^{1-k}\enspace.
  \end{equation*}
  Note that the condition required on $m$ here is implied by
  \begin{equation*}
    \frac{(k-1)(\log_2(n)-\log_2(e))-\log_2(m)}{2}\geq 1
  \end{equation*}
  and thus follows from the condition in the theorem that $\sigma\geq 1$.

  Putting this together we get average case statistical security less than or equal to $\sqrt{m(e/n)^{k-1}}$. Thus we have average case statistical security $2^{-\sigma}$ for
  \begin{equation*}
    \sigma=\frac{(k-1)(\log_2(n)-\log_2(e))-\log_2(m)}{2}\enspace.
  \end{equation*}
\end{proof}
 
\subsection{Reduction to a single input and shuffling step}
\label{sec:single_shuffling}

To analyze the average-case statistical security of $\mech{V}$ we start by upper bounding the expected total variation distance between the outputs of two executions with random inputs by a function of single random input.

\begin{lemma}
Let $\mech{V}_{k,n}$ and $\mech{V}_{k,n}'$ denote two independent executions of the $k$-parallel IKOS protocol. Then we have:
\begin{align*}
\Ex_{\tup{\rv{X}},\tup{\rv{X}}'} [\TV_{|\tup{\rv{X}},\tup{\rv{X}}'}(\mech{V}_{k,n}(\tup{\rv{X}}), \mech{V}_{k,n}(\tup{\rv{X}}'))]
&\leq
\sqrt{m^{kn-1} \Pr[\mech{V}_{k,n}(\tup{\rv{X}}) = \mech{V}_{k,n}'(\tup{\rv{X}})] - 1} \enspace.
\end{align*}
\end{lemma}
\begin{proof}
We first remove the expectation over $\tup{\rv{X}}'$ by taking its randomness inside the total variation distance (ie.\ switching from $\TV_{|\tup{\rv{X}},\tup{\rv{X}}'}$ to $\TV_{|\tup{\rv{X}}}$). Note that this is akin to a reverse Jensen inequality, and therefore we will need to pay a factor of $2$ to get the result via a triangle inequality.
The formal bound is obtained by taking $\tup{\rv{X}}''$ to be an independent copy of $\tup{\rv{X}}'$ and observing that
\begin{align*}
\Ex_{\tup{\rv{X}},\tup{\rv{X}}'} [\TV_{|\tup{\rv{X}},\tup{\rv{X}}'}(\mech{V}(\tup{\rv{X}}), \mech{V}(\tup{\rv{X}}'))]
&\leq
\Ex_{\tup{\rv{X}},\tup{\rv{X}}'} [\TV_{|\tup{\rv{X}},\tup{\rv{X}}'}(\mech{V}(\tup{\rv{X}}), \mech{V}(\tup{\rv{X}}'')) + \TV_{|\tup{\rv{X}},\tup{\rv{X}}'}(\mech{V}(\tup{\rv{X}}''), \mech{V}(\tup{\rv{X}}'))]
\\
&=
\Ex_{\tup{\rv{X}}} [\TV_{|\tup{\rv{X}}}(\mech{V}(\tup{\rv{X}}), \mech{V}(\tup{\rv{X}}''))]
+
\Ex_{\tup{\rv{X}'}} [\TV_{|\tup{\rv{X}}'}(\mech{V}(\tup{\rv{X}}''), \mech{V}(\tup{\rv{X}}'))]
\\
&=
2 \Ex_{\tup{\rv{X}}} [\TV_{|\tup{\rv{X}}}(\mech{V}(\tup{\rv{X}}), \mech{V}(\tup{\rv{X}}'))] \enspace.
\end{align*} 

Next we observe that $\tup{\rv{V}} = \mech{V}(\tup{\rv{X}}')$ has uniform distribution over the tuples in $\mathbb{G}^{kn}$ that add up to $\sum_i \rv{X}_i$.
We use this information to expand the total variation and write it as an expectation over $\tup{\rv{V}}$ as follows:
\begin{align*}
2 \TV_{|\tup{\rv{X}}}(\mech{V}(\tup{\rv{X}}), \mech{V}(\tup{\rv{X}}'))
&=
2 \TV_{|\tup{\rv{X}}}(\mech{V}(\tup{\rv{X}}), \tup{\rv{V}})
\\
&=
\sum_{\tup{v} \in \mathbb{G}^{kn}} | \Pr_{|\tup{\rv{X}}}[\mech{V}(\tup{\rv{X}}) = \tup{v}] - \Pr[\tup{\rv{V}} = \tup{v}]|
\\
&=
\sum_{\tup{v} \in \mathbb{G}^{kn} : \sum \tup{v} = \sum \tup{\rv{X}}} | \Pr_{|\tup{\rv{X}}}[\mech{V}(\tup{\rv{X}}) = \tup{v}] - m^{1-kn}|
\\
&=
m^{kn-1} \Ex_{\tup{\rv{V}}} [|\Pr_{|\tup{\rv{X}},\tup{\rv{V}}}[\mech{V}(\tup{\rv{X}}) = \tup{\rv{V}}] - m^{1-kn}|] \enspace.
\end{align*}

The final task is to bound the remaining expectation.
We start by defining the random variable $\rv{Z} = \rv{Z}(\trv{X},\trv{V}) := \Pr_{|\tup{\rv{X}},\tup{\rv{V}}}[\mech{V}(\tup{\rv{X}}) = \tup{\rv{V}}]$.
Note that because both $\mech{V}(\tup{\rv{X}})$ and $\tup{\rv{V}}$ follow the same uniform distribution over tuples in $\mathbb{G}^{kn}$ conditioned to having the same sum, we have
\begin{align*}
\Ex_{\trv{X},\trv{V}}[\rv{Z}] = \Pr[\mech{V}(\tup{\rv{X}}) = \tup{\rv{V}}] = m^{1-kn} \enspace.
\end{align*}
Therefore, the expectation that we need to bound takes the simple form $\Ex[|\rv{Z} - \Ex[\rv{Z}]|]$, and can be bounded in terms of $\Ex[\rv{Z}^2]$ via Jensen's inequality:
\begin{align*}
\Ex[|\rv{Z} - \Ex[\rv{Z}]|]
\leq
\sqrt{\Var[\rv{Z}]}
=
\sqrt{\Ex[\rv{Z}^2] - \Ex[\rv{Z}]^2} \enspace.
\end{align*}

Now recall that if $\rv{A}, \rv{A}' \in A$ are i.i.d. random variables, then we have
\begin{align*}
\Pr[\rv{A} = \rv{A}'] = \sum_{a \in A} \Pr[\rv{A} = a]^2 \enspace.
\end{align*}
Using this identity we can write the expectation of $\rv{Z}^2$ over the randomness in $\trv{V}$ in terms of the probability that two independent executions of $\mech{V}(\trv{X})$ (conditioned on $\trv{X})$ yield the same result:
\begin{align*}
\Ex_{\trv{V}}[\rv{Z}^2]
&=
m^{1-kn} \sum_{\tup{v} \in \mathbb{G}^{kn} : \sum \tup{v} = \sum \tup{\rv{X}}}
\Pr_{|\tup{\rv{X}}}[\mech{V}(\tup{\rv{X}}) = \tup{v}]^2
\\
&=
m^{1-kn} \Pr_{|\tup{\rv{X}}}[\mech{V}(\tup{\rv{X}}) = \mech{V}'(\tup{\rv{X}})]
\enspace.
\end{align*}
Putting the pieces together completes the proof:
\begin{align*}
\Ex_{\tup{\rv{X}},\tup{\rv{X}}'} [\TV_{|\tup{\rv{X}},\tup{\rv{X}}'}(\mech{V}(\tup{\rv{X}}), \mech{V}(\tup{\rv{X}}'))]
&\leq
2 \Ex_{\tup{\rv{X}}} [\TV_{|\tup{\rv{X}}}(\mech{V}(\tup{\rv{X}}), \mech{V}(\tup{\rv{X}}'))] \\
&\leq
m^{kn-1} \Ex_{\tup{\rv{X}}, \tup{\rv{V}}} [|\Pr_{|\tup{\rv{X}},\tup{\rv{V}}}[\mech{V}(\tup{\rv{X}}) = \tup{\rv{V}}] - m^{1-kn}|] \\
&\leq
\sqrt{m^{kn-1} \Pr[\mech{V}(\tup{\rv{X}}) = \mech{V}'(\tup{\rv{X}})] - 1} \enspace.
\end{align*}
\end{proof}

To further simplify the bound in previous lemma we can write the probability $\Pr[\mech{V}_{k,n}(\tup{\rv{X}}) = \mech{V}_{k,n}'(\tup{\rv{X}})]$ in terms of a single permutation step.
For that purpose we introduce the notation $\mech{V}_{k,n} = \mech{S}_{k,n} \circ \tup{\mech{R}}_{k,n}$, where:
\begin{itemize}
\item $\tup{\mech{R}}_{k,n} : \mathbb{G}^n \to \mathbb{G}^{n k}$ is the randomized map that given $\tup{x} = (x_1,\ldots,x_n)$ generates the shares $(\rv{Y}_i^{(1)}, \ldots, \rv{Y}_i^{(k)}) = \mech{R}(x_i)$ and arranges them in order first by share id and then by user:
\begin{align*}
\tup{\mech{R}}_{k,n}(\tup{x}) = (\rv{Y}_1^{(1)}, \ldots, \rv{Y}_n^{(1)}, \ldots, \rv{Y}_1^{(k)}, \ldots, \rv{Y}_n^{(k)}) \enspace.
\end{align*}
\item $\mech{S}_{k,n} : \mathbb{G}^{nk} \to \mathbb{G}^{nk}$ is a random permutation of its inputs obtained by applying $k$ independent shufflers $\mech{S}^{(j)}$, $j \in [k]$, to the inputs in blocks of $n$:
\begin{align*}
\mech{S}_{k,n}(y_1^{(1)}, \ldots, y_n^{(1)}, \ldots, y_1^{(k)}, \ldots, y_n^{(k)})
= (\mech{S}^{(1)}(y_1^{(1)}, \ldots, y_n^{(1)}) \cdots \mech{S}^{(k)}(y_1^{(k)}, \ldots, y_n^{(k)}))
\end{align*}
\end{itemize}
It is important to note that $\mech{S}_{k,n}$ produces random permutations of $[kn]$ which are uniformly distributed in the subgroup of all permutations which arise as the parallel composition of $k$ uniform permutations on $[n]$.
Equipped with these observations, it is straightforward to verify the following identity.

\begin{lemma}
Let $\tup{\mech{R}}_{k,n}$ and $\tup{\mech{R}}_{k,n}'$ denote two independent executions of the additive sharing step in $\mech{V}_{k,n}(\tup{\rv{X}}) = \mech{S}_{k,n} \circ \tup{\mech{R}}_{k,n}$.
Then we have
\begin{align*}
\Pr[\mech{V}_{k,n}(\tup{\rv{X}}) = \mech{V}'_{k,n}(\tup{\rv{X}})]
&=
\Pr[\tup{\mech{R}}_{k,n} (\tup{\rv{X}}) = \mech{S}_{k,n} \circ \tup{\mech{R}}_{k,n}' (\tup{\rv{X}})] \enspace.
\end{align*}
\end{lemma}
\begin{proof}
We drop all subscripts for convenience.
The result follows directly from the fact that $\mech{S}$ is uniform over a subgroup of permutations, which implies that the inverse of $\mech{S}$ and the composition of two independent copies of $\mech{S}$ both follow the same distribution as $\mech{S}$.
Thus, we can write:
\begin{align*}
\Pr[\mech{V}(\tup{\rv{X}}) = \mech{V}'(\tup{\rv{X}})]
&=
\Pr[\mech{S} \circ \tup{\mech{R}} (\tup{\rv{X}}) = \mech{S}' \circ \tup{\mech{R}}' (\tup{\rv{X}})]
\\
&=
\Pr[\tup{\mech{R}} (\tup{\rv{X}}) = \mech{S}^{-1} \circ \mech{S}' \circ \tup{\mech{R}}' (\tup{\rv{X}})]
\\
&=
\Pr[\tup{\mech{R}} (\tup{\rv{X}}) = \mech{S} \circ \tup{\mech{R}}' (\tup{\rv{X}})] \enspace.
\end{align*}
\end{proof}

Putting these two lemmas together yields the following bound.

\begin{lemma}\label{lem:single_shuffling}
Let $\mech{V}_{k,n}$ and $\mech{V}_{k,n}'$ denote two independent executions of the $k$-parallel IKOS protocol. Then we have:
\begin{align*}
\Ex_{\tup{\rv{X}},\tup{\rv{X}}'} [\TV_{|\tup{\rv{X}},\tup{\rv{X}}'}(\mech{V}_{k,n}(\tup{\rv{X}}), \mech{V}_{k,n}(\tup{\rv{X}}'))]
&\leq
\sqrt{m^{kn-1} \Pr[\tup{\mech{R}}_{k,n} (\tup{\rv{X}}) = \mech{S}_{k,n} \circ \tup{\mech{R}}_{k,n}' (\tup{\rv{X}})] - 1} \enspace.
\end{align*}  
\end{lemma}

\subsection{Reduction to a problem on random graphs}
\label{sec:graph_argument}

\begin{definition}
  A random $n$-vertex $2k$-regular multigraph $G$ is from the \emph{unconditioned permutation model}, denoted $R^*(n,2k)$, if it can be constructed as follows. Start with $n$ vertices and no edges. Take a set of $k$ uniformly random and independent permutations of the vertices, denoted $\{\pi_i\}_{i=1}^k$.  For each vertex $v$ and each index $i\in [k]$, add an edge between $v$ and $\pi_i(v)$. We say that $G$ is generated by $\{\pi_i\}_{i=1}^k$.
\end{definition}

Note that $G$ may have self-loops. Let $C(G)$ be number of connected components of a graph $G$.

\begin{lemma}
  \label{lem:graph_argument}
  Let $G$ be drawn from $R^*(n,2k)$, then
  \begin{equation*}
  \Pr[\tup{\mech{R}} (\tup{\rv{X}}) = \mech{S} \circ \tup{\mech{R}}' (\tup{\rv{X}})]\leq \Ex[m^{C(G)-kn}]
  \end{equation*}
\end{lemma}

\begin{proof}
  Note that, by the tower law,
  \begin{align*}
    \Pr[\tup{\mech{R}} (\tup{\rv{X}}) = \mech{S} \circ \tup{\mech{R}}' (\tup{\rv{X}})]=\Ex[\Pr[\tup{\mech{R}} (\tup{\rv{X}}) = \mech{S} \circ \tup{\mech{R}}' (\tup{\rv{X}})|\mech{S}]].
  \end{align*}
  Let $G_{\mech{S}}$ be the unconditioned permutation model graph, with vertex set $[n]$, generated by the $k$ permutations used in $\mech{S}$. Note that, it suffices to show that
  \begin{equation*}
    \Pr[\tup{\mech{R}} (\tup{\rv{X}}) = \mech{S} \circ \tup{\mech{R}}' (\tup{\rv{X}})|\mech{S}]=m^{C(G_{\mech{S}})-kn}.
  \end{equation*}

  For notational convenience, we will define a deterministic reordering of $\tup{\mech{R}} (\tup{\rv{X}})$ and  $\mech{S} \circ \tup{\mech{R}}' (\tup{\rv{X}})$ as follows. Consider the permutation $P:[kn]\rightarrow [kn]$
  \begin{equation*}
    P(j)=\left\lfloor \frac{j-1}{k}\right\rfloor+n(j-1 \textrm{ mod } k)+1.
  \end{equation*}

  Define $\rv{U},\rv{U}'\in \mathbb{G}^{kn}$ by $\rv{U}_j=\tup{\mech{R}} (\tup{\rv{X}})_{P(j)}$ and $\rv{U}'_j=P\circ \mech{S} \circ \tup{\mech{R}}' (\tup{\rv{X}})_{P(j)}$. Note that $P$ is such that the shares from each input are grouped together (in order) in $\rv{U}$. Consequently, $\rv{U}'$ groups together collections of $k$ shares, one from the output of each shuffler. Thus it suffices to show that
  \begin{equation*}
    \Pr[ \rv{U}=\rv{U}' |\mech{S}]=m^{C(G_{\mech{S}})-kn}.
  \end{equation*}
  
  For $j\in [kn]$, let $A_j$ be the event that $\rv{U}_j=\rv{U}'_j$. Now define $p_j:=\Pr[A_j|A_1,...,A_{j-1},\mech{S}]$, thus
  \begin{equation*}
    \Pr[ \rv{U}=\rv{U}' |\mech{S}]=\prod_{j=1}^{kn} p_j.
  \end{equation*}

  First we consider values of $j$ that are not divisible by $k$, i.e. they are not the final share in a group of $k$. For such a $j$, we claim $p_j=m^{-1}$. To see this, condition on $\tup{\rv{X}}$ and $\tup{\mech{R}}'$, in addition to $A_1,...,A_{j-1}$. Note that $\rv{U}_j$ and $\rv{U}_{j+1}$ only depend upon anything we've conditioned on via their sum. Therefore $\rv{U}_j$ is still uniformly distributed and has probability $m^{-1}$ of being equal to $\rv{U}'_j$.

For an index $i\in [kn]$ we define the vertex corresponding to $i$ to be the vertex $\lceil i/k \rceil$, and we define $C_i$ to be the set of vertices in the same connected component as this vertex in $G_{\mech{S}}$. For the remaining $j$'s, we distinguish the case where the corresponding vertex is the highest index in $C_j$ and the case where it isn't.

In the first case,
\begin{equation*}
  \sum_{i \textrm{ s.t. } C_i=C_j}\rv{U}'_i=\sum_{i \textrm{ s.t. } C_i=C_j}\tup{\mech{R}}' (\tup{\rv{X}})_{P(i)}
\end{equation*}
as the sums have the same summands in a different order. Further,
\begin{equation*}
  \sum_{i \textrm{ s.t. } C_i=C_j}\tup{\mech{R}}' (\tup{\rv{X}})_{P(i)}=\sum_{i \textrm{ s.t. } C_i=C_j}\rv{U}_i
\end{equation*}
as they both represent sharings of the same input values and
\begin{equation*}
  \sum_{\substack{i \textrm{ s.t. } C_i=C_j \\ i\neq j}}\rv{U}'_i=\sum_{\substack{i \textrm{ s.t. } C_i=C_j \\ i\neq j}}\rv{U}_i
\end{equation*}
as we are conditioning on $A_1,...,A_{j-1}$. Putting these together we can conclude that $p_j=1$.

For the second case, we will find that $p_j=m^{-1}$. We will show this by showing that if we condition on the value of $\rv{U}_j$ then $\rv{U}'_j$ is still uniformly distributed. That is to say that the number of possible outcomes fitting those conditions with each value of $\rv{U}_j$ is independent of that value. To show that these sets of outcomes have the same size we will partition the possible outcomes into sets of size $m$, with $\rv{U}'_j$ taking each value in $\mathbb{G}$ exactly once in each set. This will be possible because the structure of $G_{\mech{S}}$ allows us to change the value of $\rv{U}'_j$ and other values to preserve what is being conditioned on in an algebraically principled way. If $\mathbb{G}=\mathbb{Z}_m$, for some prime $m$, i.e. the set of possible outcomes forms a vector space, this can be thought of as follows. The space of possible outcomes consitent with the conditions is a subspace of the space of all outcomes. Thus showing that this subspace contains two possible values for $\rv{U}'_j$ suffices by the nice algebraic properties of vector spaces. That there is more than one possible value of $\rv{U}'_j$ is a consequence of $G_{\mech{S}}$ ``connecting the $j$th share to later shares''. The following paragraphs make this formal in the more general setting of any abelian group $\mathbb{G}$.

Consider the set $\mathcal{T}$ of choices of $(\rv{U}\cdot \rv{U}')\in \mathbb{G}^{2kn}$ that are consistent with $A_1,...,A_{j-1}$ (and a value of $\tup{\rv{X}}$). We consider the group action of $\mathbb{G}^{2kn}$ on itself by addition. We will show that, there exists a homomorphism $\mathbb{G}\rightarrow \mathbb{G}^{2kn}$ mapping $g$ to $u_g$ with the following property.  The action of $u_g$ on $\mathbb{G}^{2kn}$ fixes $\mathcal{T}$ and $\rv{U}_j$ and adds $g$ to $\rv{U}'_j$. Therefore, the equivalence relation, of being equal upto adding $u_g$ for some $g$, partitions $\mathcal{T}$ into subsets of size $m$ each containing one value for which $A_j$ holds. It follows, from the fact that each entry in $\mathcal{T}$ is equally likely, that $p_j=m^{-1}$.

To find such a homomorphism, note that there is a path in $G_S$ from the vertex corresponding to the $j$th share to a higher index vertex. This is equivalent to saying that there is a sequence $(a_1,b_1,a_2,b_2,...,a_l,b_l,a_{l+1})$ with the following properties. The $a_i$ and $b_i$ are elements of $[kn]$ and should be interpreted as indexes of $\mathbb{G}^{kn}$. For all $i\in [l]$, $\pi(b_i)=a_i$ and $b_i$ and $a_{i+1}$ correspond to the same vertex. We have $a_1=j$, $b_l>j$, $a_i\neq a_{i'}$ for any $i\neq i'$, and $b_i< j$ for all $i<l$. Let $u_g$ be the element of $\mathbb{G}^{2kn}$ with a $g$ in entries $a_2,...,a_{l+1},b_1+kn,...,b_{l}+kn$ and the identity everywhere else.

Adding $u_g$ doesn't change the truth of $A_1,...,A_{j-1}$ because $\rv{U}_{a_i}$ and $\rv{U}'_{a_i}=\tup{\mech{R}}'(\tup{\rv{X}})_{b_i}$ are always incremented together, with the exception of when $i=1$ or $l+1$ which is fine because then $a_i\geq j$. In the case of $i=1$ this adds $g$ to $\rv{U}'_j$ without changing $\rv{U}_j$. The consistency of the implied values of $\tup{\rv{X}}$ is maintained because $\rv{U}_{a_i}$ and $\tup{\mech{R}}'(\tup{\rv{X}})_{P(b_{i-1})}$ are always incremented together and affect the $\tup{\rv{X}}$ implied by $\rv{U}$ the same as that implied by $\tup{\mech{R}}'(\tup{\rv{X}})$. Thus, this $u_g$ has the properties we claimed and $p_j=m^{-1}$.

Tying this together we have that
\begin{align*}
\Pr[\rv{U} = \rv{U}']&=\Ex[\Pr[\rv{U}=\rv{U}'|\mech{S}]] \\
                                                                                  &=\Ex[\prod_{j=1}^{kn} p_j] \\
                                                                                  &=\Ex[m^{C(G_S)-kn}] \\
                                                                                  &=\Ex[m^{C(G)-kn}]\enspace.
\end{align*}
\end{proof}
 
\subsection{Understanding the number of connected components of \texorpdfstring{$G$}{G}}
\label{sec:bound_on_graph}

\begin{lemma}
  \label{lem:bound_on_graph}
   Let $n\geq 19$, $k\geq 3$ and $m\leq \frac{1}{2}\left(\frac{n}{e}\right)^{k-1}$. Let $G$ be drawn from $R^*(n,2k)$, then
  \begin{equation*}
    \PP(C(G)=c)\leq \frac{1.5^{c-1}}{c!}\left(\frac{e}{n}\right)^{(k-1)(c-1)}.
  \end{equation*}
  Therefore,
  \begin{align*}
    \Ex[m^{C(G)}]&\leq m + m^2\left(\frac{n}{e}\right)^{1-k} \enspace.
  \end{align*}
\end{lemma}

\begin{proof}
  Let $p(n,c)=\Pr[C(G^n)=c]$, we will show that the bound in the theorem holds by induction on $c$. For $c=1$ the bound is trivial as the right hand side equals $1$. For $c>1$,
  \begin{align*}
    p(n,c)&=\frac{1}{c}\sum_{S\subset [n]}\Pr[\textrm{No edge from $S$ to $[n]-S$}]p(|S|,1)p(n-|S|,c-1) \\
          &=\frac{1}{c}\sum_{s=1}^{n-c+1}\binom{n}{s}\binom{n}{s}^{-k}p(s,1)p(n-s,c-1).
  \end{align*}
  We now bound this expression, using the induction hypothesis, to find that
  \begin{align}
    \label{eq:sum_to_bound}
    p(n,c)&\leq \frac{1}{c}\sum_{s=1}^{n-c+1}\binom{n}{s}^{1-k}\frac{1.5^{c-2}}{(c-1)!}\left(\frac{e}{n-s}\right)^{(k-1)(c-2)} \\
          &= \frac{1.5^{c-1}}{c!}\left(\frac{e}{n}\right)^{(k-1)(c-1)}\frac{2e^{1-k}}{3}\sum_{s=1}^{n-c+1}\left(\frac{(n-s)!s!n^{c-1}}{n!(n-s)^{c-2}}\right)^{k-1}.
  \end{align}
  To complete the proof it suffices to show that this sum on the right is at most $1.5e^{k-1}$. Call the $s$th summand from this sum $a_s$. We separate the summands into three cases, depending on whether $s$ is greater than $n/10$ and/or less than $3n/4$. Firsty, if $s\leq n/10$, then
  \begin{align*}
    \frac{a_{s}}{a_{s-1}}&=\left(\frac{s}{n-s+1}\left(\frac{n-s+1}{n-s}\right)^{c-2}\right)^{k-1} \\
                         &\leq \left(\frac{s}{n-s}e^{\frac{c-2}{n-s}}\right)^{k-1} \\
                         &\leq \left(\frac{e^{\frac{10}{9}}}{9}\right)^{2} \\
                         &\leq \frac{1}{8}.
  \end{align*}
  Thus we can bound the early summands with a geometric series as follows.
  \begin{align*}
    \sum_{s=1}^{\lfloor n/10 \rfloor}a_s&\leq \sum_{s=1}^{\lfloor (n-c)/10 \rfloor}\frac{a_1}{8^{s-1}} \\
                                            &\leq \sum_{s=1}^\infty \frac{a_1}{8^{s-1}} \\
                                            &\leq \frac{8a_1}{7} \\
                                            &=\frac{8}{7}\left(\frac{n}{n-1}\right)^{(c-2)(k-1)} \\
                                            &\leq \frac{8}{7}e^{\frac{(c-1)(k-1)}{n}}\leq \frac{8}{7}e^{(k-1)}
  \end{align*}

  We now similarly consider the terms with $s\geq 3n/4$. For these values of $s$,
  \begin{align*}
    \frac{a_{s+1}}{a_{s}}&=\left(\frac{s+1}{n-s}\left(\frac{n-s}{n-s-1}\right)^{c-2}\right)^{k-1} \\
                         &\geq\left(\frac{s}{n-s}\right)^{k-2} \\
                         &\geq 9
  \end{align*}
  If $c>n/4$ then there are no summands for $s\geq 3n/4$. Otherwise we can bound the late summands with a geometric series as follows.
  \begin{align*}
    \sum_{s=\lceil 3n/4 \rceil}^{n-c+1}a_s&\leq \sum_{s=\lceil 3n/4 \rceil}^{n-c+1}\frac{a_{n-c+1}}{9^{n-c+1-s}} \\
                                            &\leq \sum_{s=-\infty}^{n-c+1} \frac{a_{n-c+1}}{9^{n-c+1-s}} \\
                                            &= \frac{9a_{n-c+1}}{8} \\
                                            &=\frac{9}{8}\left(\frac{(c-1)!n^{c-1}(n-c+1)!}{n!(c-1)^{c-2}}\right)^{k-1}.
  \end{align*}
  Applying Sterling's bound, $\sqrt{2\pi}n^{n+\frac{1}{2}}e^{-n}\leq n!\leq e n^{n+\frac{1}{2}}e^{-n}$, to the factorials in the above expression bounds it above by,
  \begin{equation*}
    \frac{9}{8}\left(\frac{e^2}{\sqrt{2\pi}}(c-1)^{1.5}\left(1-\frac{c-1}{n}\right)^{n-c+1.5}\right)^{k-1}.
  \end{equation*}
  As $n\geq 19$ and $c\leq n/4$, this is maximised for $c=3$, and as we also have $k\geq 3$ this results in the bound
  \begin{equation*}
    \frac{9}{8}\left(\frac{e^2}{\sqrt{2\pi}}2\sqrt{2}(1-\frac{2}{n})^{n-1.5}\right)^{k-1}\leq \left(1.27\right)^{k-1}.
  \end{equation*}
  
  Finally we consider the case of $n/10<s<3n/4$. Let $\alpha=s/n$. Substituting this into $a_s$ gives
  \begin{equation*}
    \left(\frac{((1-\alpha)n)!(\alpha n)!}{(n-1)!(1-\alpha)^{c-2}}\right)^{k-1}.
  \end{equation*}
  Applying Sterling's bound again bounds this expression by 
  \begin{align*}
    \left(\frac{e^2}{\sqrt{2\pi}}\sqrt{n}(1-\alpha)^{2.5-c+(1-\alpha)n}\alpha^{\alpha n +\frac{1}{2}}\right)^{k-1}&\leq \left(\frac{e^2\sqrt{n}}{\sqrt{2\pi}}\alpha^{\alpha n}\right)^{k-1}\enspace.
  \end{align*}
  Where the inequality holds because $(1-\alpha)\leq 1$ and, for any summand that appears in the sum, $2.5-c+(1-\alpha) n>0$.
  The final expression is maximised for $\alpha=3/4$ and there are fewer than $3n/5$ summands with $n/10<s<3n/4$. Therefore the sum of all of these terms can be bounded by,
  \begin{align*}
    \frac{3n}{5}\left(\frac{e^2\sqrt{n}}{\sqrt{2\pi}}\left(\frac{3}{4}\right)^{\frac{3n}{4}}\right)^{k-1}&\leq \left(en\left(\frac{3}{4}\right)^{\frac{3n}{4}}\right)^{k-1} \\
    &\leq 1.
  \end{align*}
  Where we have used that $k\geq 3$ and $n\geq 19$. Adding these up the sum as a whole is bounded by
  \begin{equation*}
    \frac{8}{7}e^{k-1}+1+\left(1.27\right)^{k-1}<1.5e^{k-1}.
  \end{equation*}
  To conclude the proof we consider the expectation. Below we apply the definition of expectation with the bound on the probability above.
  \begin{equation*}
    \Ex[m^{C(G)}]\leq\sum_{c=1}^{n}m^c\frac{1.5^{c-1}}{c!}\left(\frac{e}{n}\right)^{(k-1)(c-1)}
  \end{equation*}
  Notice that every term after the second is at most $\frac{me^{k-1}}{2n^{k-1}}$ times the previous term, thus
  \begin{equation*}
    \Ex[m^{C(G)}]\leq m+\frac{3m^2}{4}\left(\frac{n}{e}\right)^{1-k}\sum_{i=0}^{\infty} \left(\frac{me^{k-1}}{2n^{k-1}}\right)^i
  \end{equation*}
  Then using that $m\leq \frac{1}{2}\left(\frac{n}{e}\right)^{k-1}$ we bound the sum by $4/3$ to find
  \begin{equation*}
    \Ex[m^{C(G)}]\leq m+m^2\left(\frac{n}{e}\right)^{1-k} \enspace.
  \end{equation*}
\end{proof} 
\subsection{Reduction to random inputs}
\label{sec:random_inputs}

Finally, we prove Corollary~\ref{cor:wst-security} from Theorem~\ref{thm:avg-security} by showing that a certain level of average-case security with $k$ messages implies the same level of worst-case security with $k+1$ messages.
In addition, we show that the additional message required to reduce worst-case security to average-case security does not need to be sent through a shuffler.
Note that the expression for the required $k$ in Corollary~\ref{cor:wst-security} is a simple rearrangement of the expression for $\sigma$, so the following lemma is all that remains to be proven.

\begin{lemma}
If $\mech{V}_{k,n}$ provides average-case statistical security with parameter $\sigma$, then $\mech{V}_{k+1,n}$ and $\tilde{\mech{V}}_{k,n}$ provide worst-case statistical security with parameter $\sigma$.
\end{lemma}
\begin{proof}
Fix a pair of inputs $\tup{x}$ and $\tup{x}'$ with the same sum.
Since the output of $\mech{V}_{k+1,n}(\tup{x})$ can be simulated directly from the output of $\tilde{\mech{V}}_{k,n}(\tup{x})$ by applying a random permutation to the last $n$ elements, we have $\TV(\mech{V}_{k+1,n}(\tup{x}),\mech{V}_{k+1,n}(\tup{x}')) \leq \TV(\tilde{\mech{V}}_{k,n}(\tup{x}), \tilde{\mech{V}}_{k,n}(\tup{x}'))$, and therefore it suffices to show that $\tilde{\mech{V}}_{k,n}$ provides worst-case statistical security with parameter $\sigma$.

The key observation that allows us to reduce the worst-case security of $\tilde{\mech{V}}_{k,n}$ to the average-case security of $\mech{V}_{k,n}(\tup{x})$ is to observe that the addition of an extra share can be interpreted as adding a random value to each user's input, effectively making the inputs uniformly random.
To formalize this intuition we observe that $\mech{R}_k$ admits a recursive decomposition as follows.
Let $\rv{U} \in G$ be a uniformly random group element and $x \in G$. Then we have $\mech{R}_1(x) = x$ and, for $k \geq 1$,
\begin{align}
\mech{R}_{k+1}(x) = (\mech{R}_{k}(x - \rv{U}), \rv{U}) \enspace.
\end{align}
Expanding this identity into the definition of $\tilde{\mech{V}}$ and writing $\tup{\rv{U}} = (\rv{U}_1, \ldots, \rv{U}_n)$ for the uniform random variables arising from applying the above expression for $\mech{R}_{k+1}$ to the input from each user, we obtain
\begin{align}
\tilde{\mech{V}}_{k}(\tup{x})
=
(\mech{V}_{k}(\tup{x} - \tup{\rv{U}}), \tup{\rv{U}}) \enspace.
\end{align}
Note that here $\tup{x} - \tup{\rv{U}}$ is a uniform random vector in $G^n$.
The result now follows from matching the uniform randomness from $\tup{\rv{U}}$ observed when executing the protocol with two inputs with the same sum:
\begin{align}
\TV(\tilde{\mech{V}}_{k+1}(\tup{x}), \tilde{\mech{V}}_{k+1}(\tup{x}'))
&=
\TV((\mech{V}_{k}(\tup{x} - \tup{\rv{U}}), \tup{\rv{U}}), (\mech{V}_{k}(\tup{x}' - \tup{\rv{U}}'), \tup{\rv{U}}'))
\\
&=
\Ex_{\tup{\rv{U}}} [\TV_{|\tup{\rv{U}}}(\mech{V}_{k}(\tup{x} - \tup{\rv{U}}), \mech{V}_{k}(\tup{x}' - \tup{\rv{U}}))]
\\
&=
\Ex_{\tup{\rv{X}},\tup{\rv{X}}'} [\TV_{|\tup{\rv{X}},\tup{\rv{X}}'}(\mech{V}_{k}(\tup{\rv{X}}), \mech{V}_{k}(\tup{\rv{X}}'))]
\enspace,
\end{align}
where $\tup{\rv{X}}$ and $\tup{\rv{X}}'$ are tuples with $n$ uniformly random group elements conditioned on $\sum_i \rv{X}_i = \sum_i \rv{X}_i'$.
\end{proof}

\bibliographystyle{plain}
\bibliography{main.bib}

\end{document}